\documentclass[letterpaper, 10 pt]{IEEEtran}
\IEEEoverridecommandlockouts
\usepackage{cite}
\usepackage{amsmath,amssymb,amsfonts}
\usepackage{algorithmic}
\usepackage{graphicx}
\usepackage{textcomp}
\usepackage{xcolor}
\usepackage{mathrsfs}
\usepackage{psfrag}
\usepackage{cancel}
\usepackage{verbatim}
\usepackage{graphicx}
\usepackage{graphicx}
\usepackage{soul}
\usepackage{circuitikz}
\usepackage[percent]{overpic}
\usetikzlibrary{decorations.pathreplacing,calligraphy,patterns}
\makeatletter
\let\NAT@parse\undefined
\makeatother
\usepackage{hyperref}
\usepackage{theoremref}

\def\mc{\mathcal}
\def\mb{\mathbb}
\def\l({\left(}
\def\r){\right)}
\def\mi{\mathrm{int}}

\begin{document}
\sloppy
\title{\LARGE \bf Digital Control of Negative Imaginary Systems Using Discrete-Time Multi-HIGS: Application to a Dual-Stage MEMS Force Sensor}
\author{Kanghong Shi, Diyako Dadkhah, Ian R. Petersen, S. O. Reza Moheimani
\thanks{This work was supported by the Australian Research Council under grant DP230102443. This work was supported in part by the University of Texas at Dallas through the James Von Ehr Distinguished Chair in Science and Technology. \emph{(Corresponding author: S. O. Reza Moheimani.)}}
\thanks{Kanghong Shi is with the Department of Electrical and Computer Systems Engineering, Monash University, VIC, Australia. Diyako Dadkhah and S. O. Reza Moheimani are with the Erik Jonsson School of Engineering and Computer Science, The University of Texas at Dallas, Richardson, TX 75080 USA. Ian R. Petersen is with the School of Engineering, College of Systems and Society, Australian National University, Canberra, Acton, ACT 2601, Australia. {\tt kanghong.shi@outlook.com}, {\tt diyako.dadkhah@utdallas.edu}, {\tt ian.petersen@anu.edu.au}, {\tt reza.moheimani@utdallas.edu}.}
}

\newtheorem{definition}{Definition}
\newtheorem{theorem}{Theorem}
\newtheorem{conjecture}{Conjecture}
\newtheorem{lemma}{Lemma}
\newtheorem{remark}{Remark}
\newtheorem{corollary}{Corollary}
\newtheorem{assumption}{Assumption}

\renewcommand{\IEEEQED}{\IEEEQEDopen}

\maketitle
\thispagestyle{plain}
\pagestyle{plain}

\begin{abstract}
In this paper, we propose a digital control approach for multi-input multi-output negative imaginary (NI) systems using discrete-time hybrid integrator–gain systems (HIGS) controllers. We show the NI property of the bimodal and trimodal discrete-time HIGS, as well as the parallel combinations of them, which are referred to as the multi-HIGS. Also, we demonstrate that linear NI systems can be asymptotically stabilized using discrete-time HIGS in digital control. We apply discrete-time bimodal and trimodal multi-HIGS controllers to a two-input two-output dual-stage force sensor with lightly damped resonant modes. To validate the theoretical findings, the closed-loop performance is evaluated in both time and frequency domains. Experimental results show that the discrete-time multi-HIGS effectively suppresses resonances while preserving favorable phase characteristics, which highlights its potential as a robust nonlinear NI controller for the digital control of NI systems.
\end{abstract}


\section{INTRODUCTION}

Linear time-invariant (LTI) feedback control forms the cornerstone of modern control system design due to its rigorous theoretical foundations and the availability of powerful analysis and synthesis tools. However, the performance achievable with linear controllers is fundamentally constrained by intrinsic limitations such as the Bode gain–phase relationship, fixed system dynamics, and inherent performance trade-offs \cite{middleton1991trade}. These limitations become particularly pronounced in high-precision control systems for plants with lightly damped resonant dynamics, where improving low-frequency disturbance rejection often leads to a reduction in phase margin and increased sensitivity to resonant modes. As a result, purely linear control architectures may face inherent difficulties in simultaneously achieving high tracking accuracy, disturbance attenuation, and robust stability in lightly-damped precision mechatronic systems.

Nonlinear and hybrid control strategies have been investigated extensively \cite{dong2008control, peng2005modeling, ouakad2015nonlinear, luo2016chaos, liguo2004hybrid} to introduce additional design flexibility. A broad class of nonlinear control relies on state-dependent switching mechanisms, where the control action changes according to the evolution of the system trajectory. Examples include hysteresis control \cite{weller2002hysteresis}, relay control \cite{stefanini2011miniature}, and sliding mode control \cite{zhang2016adaptive}, which employ discontinuous control laws triggered when the system state crosses prescribed switching surfaces. 

Hybrid control forms an important subclass of switching-based strategies in which continuous controller dynamics are coupled with discrete switching logic. Notable examples are the Clegg integrator \cite{clegg1958nonlinear}, first-order reset element \cite{horowitz1975non}, 
and related hybrid structures whose dynamics change according to the system trajectory. The hybrid integrator gain system (HIGS) is a type of nonlinear hybrid controller that combines the properties of an integrator and a static gain within a sector bounded framework. The controller switches between integrator and gain modes depending on the input-output trajectory relative to a predefined sector \cite{deenen2017hybrid}. This mechanism preserves the high low-frequency gain of an integrator while significantly reducing the associated phase lag. In contrast to linear integral action which introduces $90$ degrees of phase lag, a HIGS only introduces approximately $38.15$ degrees of phase lag while producing a continuous control signal. These characteristics make the HIGS an attractive choice for high performance motion \cite{dadkhah2025digital, van2020experimental} and vibration \cite{shi2024hybrid, van2020hybrid, heertjes2019hybrid} control of lightly damped mechanical systems where both strong disturbance rejection and adequate phase margin are required.

Implementation of hybrid controllers in practical systems is typically carried out using digital hardware. When controllers are designed in continuous-time and subsequently discretized, the resulting sampled implementation may not accurately reproduce the intended hybrid dynamics due to sampling effects and limited interface bandwidth, specially in high-bandwidth systems or systems with multiple resonant modes. Moreover, the stability of such implementations often requires sufficiently high sampling frequencies \cite{franklin1998digital}. This may exceed the capabilities of available control hardware. These considerations motivate the direct development of discrete-time formulations of the HIGS, in which the hybrid switching mechanism is defined explicitly in the discrete domain \cite{sharif2022discrete}.

The discrete-time HIGS must ensure that the input-output trajectory remains within the prescribed sector under finite sampling intervals. Early formulations \cite{shi2025discrete} employ a bimodal structure that switches between integrator and gain modes depending on whether the predicted integrator update remains inside the sector \cite{sharif2022discrete}. However, the predicted integrator dynamics may leave the sector through either of its two boundaries, namely $u = k_h e$ or $u=0$, whereas the bimodal HIGS always restores the output to the gain boundary $u = k_h e$. To address this, the trimodal HIGS is introduced in \cite{sharif2024analysis}, which has a third operating mode (a zeroing mode) to project the integrator dynamics to the boundary through which the candidate trajectory exits the sector. This additional mode helps to more accurately capture the intended sector-constrained behavior and prevents large instantaneous jumps in the HIGS output.

The stability of HIGS can be analyzed using negative imaginary (NI) systems theory \cite{lanzon2008stability,petersen2010feedback}, which provides an effective approach for the robust control of mechatronic structures with highly resonant dynamics. The NI property for a continuous-time system is satisfied if the system is dissipative with respect to the supply rate $u^\top \dot y$, where $u$ and $y$ are its input and output, respectively \cite{ghallab2025negative}. Such a system can be stabilized using another NI system in positive feedback, given that either of them has some strictness property \cite{shi2023output}. This continuous-time NI property can be preserved under a zero-order hold (ZOH) discretization process in the sense that the corresponding discrete-time property, which is referred to as the ZOH-NI property, is fully expressed in terms of system variables at sampling instants \cite{shi2024discrete}. This observation, together with the inherent physical characteristics of the NI property, motivates a digital control framework where a discrete-time controller is directly designed for the ZOH-sampled NI plant \cite{shi2024discrete}. Such a control framework supports the use of the discrete-time HIGS as an NI controller \cite{shi2024digital}, where the HIGS element is designed directly from the sampled-data model without the need to first synthesize a continuous-time HIGS and subsequently discretize it. Compared with the conventional continuous-time design route, this approach is more direct, imposes less restrictive sampling-rate requirements, and avoids mismatch between the intended continuous-time design and its digital implementation.

Existing studies on HIGS-based NI control have mainly focused on continuous-time settings \cite{shi2022negative,shi2023negative} or on single-input single-output (SISO) discrete-time systems \cite{shi2025discrete}, which limits their applicability to high-speed systems with multivariable dynamics. In this paper, we apply discrete-time multi-HIGS for the digital control of multi-input multi-output (MIMO) NI plants. Specifically, we investigate the NI properties of bimodal and trimodal single-HIGS in discrete-time and extend them for multi-HIGS. Theoretical analysis is presented to establish the asymptotic stability of single-HIGS and multi-HIGS configurations for both bimodal and trimodal implementations. This framework highlights the practical applicability of discrete-time multi-HIGS for the control of NI systems equipped with collocated force actuators and position sensors. As a demonstration, the control framework is applied to a dual-stage microelectromechanical system (MEMS) force sensing device reported in \cite{dadkhah2024design}, where two parallel HIGS elements are designed to provide resonance damping, stabilize the plant, and mitigate cross-coupling between the stages. Frequency- and time-domain results are also presented to validate the theoretical findings.

The remainder of this article is organized as follows. Section \ref{sec:pre} presents preliminaries on HIGS and NI systems. Section \ref{Bimodal discrete-time HIGS and NI control} introduces the discrete-time bimodal single-HIGS and multi-HIGS models and establishes their SANI property. Similarly, Section \ref{Trimodal discrete-time HIGS and NI control} presents discrete-time trimodal single-HIGS and multi-HIGS models along with their SANI property. Both Sections \ref{Bimodal discrete-time HIGS and NI control} and \ref{Trimodal discrete-time HIGS and NI control} further demonstrate that discrete-time bimodal and trimodal multi-HIGS structures, composed of single-HIGS elements connected in parallel, can asymptotically stabilize linear NI multivariable systems when implemented in positive feedback. In Section \ref{Application: A MEMS Force Sensor}, the proposed discrete-time bimodal and trimodal multi-HIGS structures are experimentally applied to the control of a two-input two-output MEMS force sensor, and their performance is evaluated using frequency- and time-domain results. Finally, Section \ref{Conclusion} concludes the article.

Notation: The notation in this paper is standard. $\mathbb R$ denotes the field of real numbers. $\mb N$ denotes the set of nonnegative integers. $\mathbb R^{m\times n}$ denotes the space of real matrices of dimension $m\times n$. $A^{\top}$ denotes the transpose of a matrix $A$.  $A^{-\top}$ denotes the transpose of the inverse of $A$; that is, $A^{-\top}=(A^{-1})^{\top}=(A^{\top})^{-1}$. $\lambda_{max}(A)$ denotes the largest eigenvalue of a matrix $A$ with real spectrum. $\|\cdot\|$ denotes the standard Euclidean norm. For a real symmetric or complex Hermitian matrix $P$, $P>0\ (P\geq 0)$ denotes the positive (semi-)definiteness of a matrix $P$ and $P<0\ (P\leq 0)$ denotes the negative (semi-)definiteness of a matrix $P$. A function $V: \mb R^n \to \mb R$ is said to be positive definite if $V(0)=0$ and $V(x)>0$ for all $x\neq 0$. Let $\theta_i\in \mb R^p$ denote
the standard unit vector; i.e., the $i$-th element of $\theta_i$ is one and all other elements are zeros.

\section{PRELIMINARIES}\label{sec:pre}
In this section, we review some preliminary results on ZOH-NI systems that are introduced in \cite{shi2024discrete}. Such systems naturally arise when an NI physical model undergoes the sample and hold process in digital control, in which typical ZOH devices are used. Consider the discrete-time system
\begin{subequations}\label{eq:DT_nonlinear}
\begin{align}
	x_{k+1} =&\ f(x_k,u_k),\label{eq:state eq}\\
	y_k=&\ h(x_k),\label{eq:output eq}
\end{align}	
\end{subequations}
where $x_k \in \mc D \subseteq \mb R^n$, $u_k,y_k \in \mathbb R^p$ are the state, input and output of the system, respectively, at time step $k\in \mathbb N$. Here, $\mc D$ is an open set such that $0\in \mc D$. Also, $f\colon\mc D \times \mathbb R^p\to \mc D$ is Lipschitz in $x$, and $h\colon\mc D \to \mathbb R^p$ is continuous. We assume $f(0,0)=0$ and $h(0)=0$, without loss of generality.

\begin{definition}[ZOH-NI]\label{def:DT_NNI}\cite{shi2024discrete}
The system (\ref{eq:DT_nonlinear}) is said to be a zero-order hold negative imaginary (ZOH-NI) system if there exists a positive definite function $V\colon \mc D \to \mb R$ such that for arbitrary $x_k$ and $u_k$,
\begin{equation}\label{eq:NNI ineq}
V(x_{k+1})-V(x_{k})\leq u_k^{\top}\left(y_{k+1}-y_{k}\right),	
\end{equation}
for all $k$.
\end{definition}

We next recall a necessary and sufficient LMI condition under which Definition \ref{def:DT_NNI} is satisfied by a linear system of the form
\begin{subequations}\label{eq:G(z)}
	\begin{align}
\Sigma\colon\ 		x_{k+1} =&\ Ax_k+Bu_k,\label{eq:G(z) state eq}\\
		y_k =&\ Cx_k,\label{eq:G(z) output eq}
	\end{align}
\end{subequations}
where $u_k,y_k\in \mb R^p$ and $x_k\in \mathbb R^n$ are the input, output, and state of the system, respectively.
\begin{lemma}\label{lemma:LMI new DT-NI}\cite{shi2024discrete}
Suppose the linear system (\ref{eq:G(z)}) satisfies $\det(I-A)\neq 0$. Then the system (\ref{eq:G(z)}) is ZOH-NI with a positive definite quadratic storage function satisfying (\ref{eq:NNI ineq}) if and only if there exists a real matrix $P=P^{\top}>0$ such that
\begin{equation*}
	A^{\top}PA-P\leq 0 \quad \textnormal{and} \quad C = B^{\top}(I-A)^{-\top}P.
\end{equation*}
\end{lemma}

To stabilize a ZOH-NI system, \cite{shi2024discrete} provides a control framework where a step-advanced negative imaginary (SANI) controller is applied in positive feedback to a ZOH plant. We show later in the present article that discrete-time HIGS are qualified as such controllers. Here, we provide the definition of SANI systems in the following. Consider the system
\begin{subequations}\label{eq:nonlinear SANI system}
\begin{align}
\widetilde x_{k+1} =&\ \widetilde f(\widetilde x_k, \widetilde u_k),\\
	\widetilde y_k=&\ \widetilde h(\widetilde x_k, \widetilde u_k),
\end{align}	
\end{subequations}
where $\widetilde x_k \in \widetilde{\mc D} \subseteq \mb R^m$, $\widetilde u_k, \widetilde y_k \in \mathbb R^p$ are the state, input and output of the system, respectively, at time step $k\in \mathbb N$. Here, $\mc D$ is an open set such that $0\in \widetilde{\mc D}$. Also, $\widetilde f\colon \widetilde{\mc D} \times \mathbb R^p\to \widetilde{\mc D}$ is Lipschitz in $\widetilde x$ and $\widetilde h\colon \widetilde{\mc D} \times \mb R^p \to \mathbb R^p$ is continuous. We assume $\widetilde f(0,0)=0$ and $\widetilde h(0,0)=0$, without loss of generality.
\begin{definition}[SANI]\label{def:SANI}\cite{shi2024discrete}
	The system \eqref{eq:nonlinear SANI system} is said to be a step-advanced negative imaginary (SANI) system if there exists a continuous function $\widehat h(x_k)$ such that:
\begin{enumerate}
	\item $\widetilde h(\widetilde x_k, \widetilde u_k)=\widehat h(\widetilde f(\widetilde x_k, \widetilde u_k))$;
	\item there exists a continuous positive definite function $\widetilde V\colon \mc D \to \mb R$ such that for arbitrary state $\widetilde x_k$ and input $\widetilde u_k$,
\begin{equation*}
\widetilde V(\widetilde x_{k+1})-\widetilde V(\widetilde x_{k})\leq \widetilde u_k^{\top}\left(\widehat h(\widetilde x_{k+1})-\widehat h(\widetilde x_{k})\right)
\end{equation*}
for all $k$.
\end{enumerate}
\end{definition}
\begin{remark}\label{remark:NI_SANI}
	Definition \ref{def:SANI} can be regarded as a variant of Definition \ref{def:DT_NNI} in which the system output is taken one step in advance. Specifically, suppose the system (\ref{eq:DT_nonlinear}) is ZOH-NI as per Definition \ref{def:DT_NNI}. Then a system with the same state equation (\ref{eq:state eq}) and an output equation $\widetilde y_k = \widehat h(x_{k+1})=\widetilde h(f(x_k,u_k))$ is an SANI system. This does not affect causality because $\widetilde h(f(x_k,u_k))$ is a function of the state $x_k$ and input $u_k$ at the current step $k$.
\end{remark}

\section{BIMODAL DISCRETE-TIME HIGS AND NI CONTROL}\label{Bimodal discrete-time HIGS and NI control}
In this section, we investigate NI control using a discrete-time bimodal HIGS. First, we provide the system model for a single bimodal HIGS and demonstrate that it possesses the SANI property. Next, we introduce the bimodal multi-HIGS architecture and prove its SANI property. Finally, we establish the closed-loop stability for the interconnection of a MIMO ZOH-NI plant and a bimodal multi-HIGS controller.
\subsection{Bimodal single HIGS model and SANI property} \label{Bimodal single HIGS model and SANI property}
The discrete-time bimodal HIGS was originally introduced in \cite{sharif2022discrete}. Here, we adapt the model presented in \cite{sharif2022discrete} to align with the nonlinear state-space description in \eqref{eq:DT_nonlinear}. The discrete-time bimodal HIGS is given by:
\begin{equation}\label{eq:bimodal HIGS}
\mathcal H:
\begin{cases}
x_h(k+1) = x_{\mathrm{int}}(k), & \text{if } (e(k),x_{\mathrm{int}}(k)) \in \mathcal F,\\[0.5mm]
x_h(k+1) = \kappa_h e(k),       & \text{if } (e(k),x_{\mathrm{int}}(k)) \notin \mathcal F,\\[0.5mm]
y_h(k) = x_h(k+1),
\end{cases}
\end{equation}
where $e(k),y_h(k),x_h(k)\in \mb R$ are the system input, output, and state, respectively. The constant parameters $\omega_h \geq 0$ and $\kappa_h > 0$ denote the integrator frequency and the gain value, respectively. The quantity
\begin{equation}\label{eq:x int}
x_{\mathrm{int}}(k) = x_h(k) + \omega_h e(k)
\end{equation}
represents the candidate state value obtained by applying linear integral action to the input $e(k)$ from the current state $x_h(k)$. The HIGS is designed to operate primarily in the integrator mode. This mode is active as long as the candidate value $x_{\mathrm{int}}(k)$ remains within the sector $[0,\kappa_h]$ relative to the input $e(k)$. Specifically, the integrator region $\mathcal F$ is defined as:
\begin{equation}\label{eq:F}
\mathcal F
= \bigl\{(e(k),x_{\mathrm{int}}(k))\in\mb R^2
\;\big|\;
x_{\mathrm{int}}(k)e(k) \ge \tfrac{1}{\kappa_h} x_{\mathrm{int}}(k)^2\bigr\}.
\end{equation}
If the candidate trajectory tends to leave this sector; i.e., $(e(k),x_{\mathrm{int}}(k)) \notin \mathcal F$, the HIGS switches to the gain mode, updating the state to $x_h(k+1) = \kappa_h e(k)$. This projection ensures the output strictly satisfies the sector constraint. Because the constraint is enforced directly on the input-output pair $(e(k),y_h(k)) = (e(k),x_h(k+1))$, we have $(e(k),y_h(k)) \in \mathcal F$ for all $k \ge 1$, regardless of the initial condition $x_h(0)$.

To streamline the notation in the subsequent analysis, we will denote the time-dependent variables $e(k)$, $x_h(k)$, $y_h(k)$, and $x_{\mathrm{int}}(k)$ using subscripts as $e_k$, $\widetilde x_k$, $\widetilde y_k$, and $\widetilde x_{k|\mathrm{int}}$, respectively. Additionally, note that the parameter $\omega_h$ used in this paper corresponds to the product $\omega_h T_s$ in \cite{sharif2022discrete,sharif2024analysis}, where $T_s$ is the sampling period. Because our analysis is conducted entirely in the discrete-time domain, we absorb $T_s$ into the integrator frequency $\omega_h$.

As established in \cite{shi2024digital}, the discrete-time bimodal HIGS is an SANI system. In other words, the system \eqref{eq:bimodal HIGS} exhibits the ZOH-NI property from the input $e_k$ to the state $\widetilde x_k$.

\begin{theorem}\cite{shi2024digital}\label{theorem:single HIGS SANI}
The bimodal HIGS of the form (\ref{eq:bimodal HIGS}) is an SANI system with the storage function
\begin{equation}\label{eq:HIGS storage function}
	\widetilde V(\widetilde x_k) = \frac{1}{2\kappa_h}\widetilde x_k^2
\end{equation}
satisfying
\begin{equation}\label{eq:NNI ineq for HIGS aux}
	\widetilde V(\widetilde x_{k+1})-\widetilde V(\widetilde x_k)\leq e_k(\widetilde x_{k+1}-\widetilde x_k),
\end{equation}
for any input $e_k$ and state $\widetilde x_k$.
\end{theorem}

\subsection{Multi-HIGS structure and SANI property}
We now introduce the discrete-time multi-HIGS architecture, which serves as the foundational MIMO framework for both the bimodal and trimodal HIGS controllers discussed in subsequent sections.

\begin{figure}[h!]
\centering
\psfrag{e_1}{$e^{[1]}$}
\psfrag{e_2}{$e^{[2]}$}
\psfrag{e_n}{$e^{[p]}$}
\psfrag{udd}{$\vdots$}
\psfrag{u_1}{$\widetilde y^{[1]}$}
\psfrag{u_2}{$\widetilde y^{[2]}$}
\psfrag{u_n}{$\widetilde y^{[p]}$}
\psfrag{odd}{$\vdots$}
\psfrag{ddd}{$\vdots$}
\psfrag{H_p1}{$\mc H_1$}
\psfrag{H_p2}{$\mc H_2$}
\psfrag{H_pn}{$\mc H_p$}
\psfrag{H_0}{$\widehat {\mc H}$}
\includegraphics[width=9.5cm]{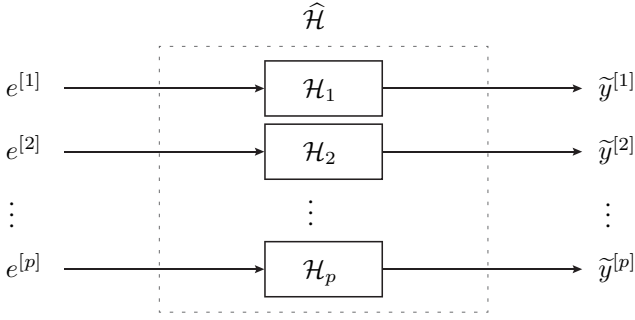}
\caption{A discrete-time multi-HIGS $\widehat {\mc H}$ formed by connecting $p$ scalar HIGS $\mc H_i$ in parallel. It is called a bimodal (resp. trimodal) multi-HIGS if each $\mc H_i$ is a bimodal (resp. trimodal) HIGS.}
\label{fig:multi_HIGS}
\end{figure}

Consider $p$ scalar HIGS $\mc H_i$ connected in parallel, as shown in Fig.~\ref{fig:multi_HIGS}. The input, output and state of the $i$-th HIGS $\mc H_i$ at time step $k\in \mb N$ are denoted by $e^{[i]}_k$, $\widetilde y^{[i]}_k$ and $\widetilde x^{[i]}_k$, respectively. The stacked input, output and state vectors of the multi-HIGS $\widehat {\mc H}$ are given by
\begin{align}
	E_k =& \begin{bmatrix} 
 	e^{[1]}_k& e^{[2]}_k &\cdots & e^{[p]}_k
 \end{bmatrix}^{\top},\label{eq:multi_HIGS_input}\\
	\widetilde Y_k =& \begin{bmatrix} 
 	\widetilde y^{[1]}_k & \widetilde y^{[2]}_k&\cdots & \widetilde y^{[p]}_k
 \end{bmatrix}^{\top},\label{eq:multi_HIGS_output}\\
\widetilde X_k =& \begin{bmatrix} 
 	\widetilde x^{[1]}_k& \widetilde x^{[2]}_k& \cdots & \widetilde x^{[p]}_k
 \end{bmatrix}^{\top}.\label{eq:multi_HIGS_state}
\end{align}

Each individual HIGS $\mc H_i$ is parameterized by an integrator frequency $\omega_i \geq 0$ and a gain value $\kappa_i > 0$. We compactly represent these gain values using the diagonal matrix:
\begin{equation}\label{eq:K_def}
	K = \mathrm{diag}\{\kappa_1,\cdots,\kappa_p\}.
\end{equation}

Before specifying the bimodal or trimodal cases, we establish a general result demonstrating that any multi-HIGS constructed from parallel SANI scalar channels inherently preserves the overall SANI property.

\begin{theorem}\label{theorem:SANI property of multi-HIGS}
Consider a multi-HIGS $\widehat{\mc H}$ as shown in Fig.~\ref{fig:multi_HIGS}, where the $i$-th scalar channel $\mc H_i$ has the storage function:
\begin{equation}\label{eq:HIGS i storage generic}
\widetilde V_i\left(\widetilde x^{[i]}_k\right) = \frac{1}{2\kappa_i}\left(\widetilde x^{[i]}_k\right)^2
\end{equation}
and satisfies the SANI dissipation inequality:
\begin{equation}\label{eq:SANI ineq for HIGS i generic}
		\widetilde V_i\left(\widetilde x^{[i]}_{k+1}\right)-\widetilde V_i\left(\widetilde x^{[i]}_k\right)\leq e^{[i]}_k\left(\widetilde x^{[i]}_{k+1}-\widetilde x^{[i]}_k\right).
\end{equation}
Then the multi-HIGS $\widehat{\mc H}$ is an SANI system from the input $E_k$ to the output $\widetilde Y_k$ with the storage function
	\begin{equation}\label{eq:multi-HIGS storage}
		\widehat V(\widetilde X_k) = \frac{1}{2}\widetilde X_k^{\top} K^{-1}\widetilde X_k
	\end{equation} 
	satisfying
	\begin{equation}\label{eq:multi-HIGS SANI inequality}
\widehat V(\widetilde X_{k+1})-	\widehat V(\widetilde X_{k})\leq E_k^{\top}(\widetilde X_{k+1} - \widetilde X_{k}),
	\end{equation}
where $K$ is given by \eqref{eq:K_def}, and $\widetilde Y_k = \widetilde X_{k+1}$.
\end{theorem}
\begin{IEEEproof}
By definition, the storage function of the multi-HIGS given by \eqref{eq:multi-HIGS storage} is the sum of all individual HIGS's storage functions \eqref{eq:HIGS i storage generic}; i.e.,
\begin{equation*}
	\widehat V(\widetilde X_k) = \sum_{i=1}^p\widetilde V_i\left(\widetilde x^{[i]}_k\right).
\end{equation*}
Applying the vector definitions \eqref{eq:multi_HIGS_input} and \eqref{eq:multi_HIGS_state} alongside the individual dissipation inequalities \eqref{eq:SANI ineq for HIGS i generic}, we obtain:
\begin{align}
	\widehat V(\widetilde X_{k+1})-	\widehat V(\widetilde X_{k})=& \sum_{i=1}^p\left[\widetilde V_i\left(\widetilde x^{[i]}_{k+1}\right)-\widetilde V_i\left(\widetilde x^{[i]}_k\right)\right]\notag\\
	\leq & \sum_{i=1}^p e^{[i]}_k\left(\widetilde x^{[i]}_{k+1}-\widetilde x^{[i]}_k\right)\notag\\
	= &\ E_k^{\top}(\widetilde X_{k+1} - \widetilde X_{k})\notag
\end{align}
which is \eqref{eq:multi-HIGS SANI inequality}. Since by construction $\widetilde Y_k=\widetilde X_{k+1}$, the multi-HIGS is SANI according to Definition \ref{def:SANI}.
\end{IEEEproof}

\subsection{Bimodal multi-HIGS model}
When each scalar channel $\mc H_i$ in Fig.~\ref{fig:multi_HIGS} is a bimodal HIGS of the form \eqref{eq:bimodal HIGS}, the resulting composite architecture is termed a bimodal multi-HIGS. The state-space representation of this bimodal multi-HIGS is given by:
\begin{subequations}\label{eq:bimodal multi-HIGS}
	\begin{align}
e^{[i]}_k =&\ \theta_i^{\top}E_{k},\\
\widetilde {x}^{[i]}_{k+1} =&\ \widetilde x^{[i]}_{k|\mathrm{int}}, &\text{if}\, (e^{[i]}_k,\widetilde x^{[i]}_{k|\mathrm{int}}) \in \mathcal{F}_i\label{eq:bi-HIGS i integrator mode}\\
\widetilde x^{[i]}_{k+1} =&\ \kappa_i e^{[i]}_k, & \text{if}\, (e^{[i]}_k,\widetilde x^{[i]}_{k|\mathrm{int}}) \notin \mathcal{F}_i\label{eq:bi-HIGS i gain mode}\\
\widetilde X_k =&\ \left[\begin{matrix}
	\widetilde x^{[1]}_k, \widetilde x^{[2]}_k,\cdots, \widetilde x^{[p]}_k
			\end{matrix}\right]^{\top},\\
\widetilde Y_k =&\ \widetilde X_{k+1},\label{eq:bimodal multi-HIGS output}
	\end{align}
\end{subequations}
where $\theta_i\in \mb R^p$ represents the $i$-th standard basis vector, $\widetilde x^{[i]}_{k|\mathrm{int}}=\widetilde x^{[i]}_k+ \omega_{i}e^{[i]}_k$ is the candidate integrator state for the $i$-th channel, and the sector boundary set $\mc F_i$ is defined as:
\begin{align}\label{eq:F_i}
	\mc F_i =& \Bigl\{(e,x_{\mathrm{int}}) \in \mb R^2\Bigm|
	x_{\mathrm{int}} e\geq \tfrac{1}{\kappa_i}x_{\mathrm{int}}^2\Bigr\}.
\end{align}

For each bimodal HIGS $\mc H_i$, Theorem~\ref{theorem:single HIGS SANI} implies that the storage function \eqref{eq:HIGS storage function} with $\kappa_h$ replaced by $\kappa_i$ satisfies \eqref{eq:SANI ineq for HIGS i generic}. Hence Theorem~\ref{theorem:SANI property of multi-HIGS} implies directly that the bimodal multi-HIGS \eqref{eq:bimodal multi-HIGS} is SANI with storage function \eqref{eq:multi-HIGS storage}.

\begin{corollary}
The bimodal multi-HIGS given by \eqref{eq:bimodal multi-HIGS} is an SANI system from the input $E_k$ to the output $\widetilde Y_k$, with the storage function \eqref{eq:multi-HIGS storage}.
\end{corollary}
\begin{IEEEproof}
This result follows directly from Theorems~\ref{theorem:single HIGS SANI} and \ref{theorem:SANI property of multi-HIGS}.
\end{IEEEproof}

\subsection{Control of NI systems using bimodal multi-HIGS}
We now show that a linear ZOH-NI plant can be stabilized by applying a suitable bimodal discrete-time multi-HIGS controller in positive feedback. 
Consider the interconnection of a linear ZOH-NI system and a multi-HIGS as shown in Fig.~\ref{fig:MIMO interconnection}, with the settings
\begin{align}
	u_k =&\ \widetilde Y_k;\label{eq:setting u=tilde Y}\\
	E_k =&\ y_k.\label{eq:setting E=y}
\end{align}

\begin{figure}[h!]
\centering
\psfrag{in_1}{$u_k$}
\psfrag{y_1}{$y_k$}
\psfrag{e}{\hspace{-0.1cm}$E_k$}
\psfrag{x_h}{$\widetilde Y_k$}
\psfrag{plant}{\hspace{-0.1cm}$G(z)$}
\psfrag{HIGS}{\hspace{0.25cm}$\widehat{\mc H}$}
\includegraphics[width=8.5cm]{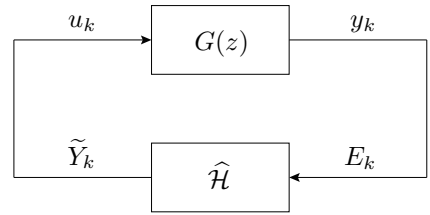}
\caption{Closed-loop interconnection of a linear ZOH-NI system given by (\ref{eq:G(z)}) with transfer matrix $G(z)$ and a discrete-time multi-HIGS $\widehat {\mc H}$.}
\label{fig:MIMO interconnection}
\end{figure}

\begin{theorem}\label{theorem:bimodal stability}
	Consider a system with transfer matrix $G(z)$ and a minimal realization (\ref{eq:G(z)}) where $\det(I-A)\neq 0$. Suppose $G(z)$ is ZOH-NI with a positive definite quadratic storage function. Also, consider a bimodal multi-HIGS $\widehat {\mc H}$ of the form (\ref{eq:bimodal multi-HIGS}) such that $0<\omega_i\leq \kappa_i$ for all $i\in \{1,2,\cdots,p\}$ and $K^{-1}-G(1)>0$ where $K$ is given by (\ref{eq:K_def}). Then the closed-loop interconnection of $G(z)$ and the bimodal multi-HIGS $\widehat {\mc H}$ is asymptotically stable.
\end{theorem}
\begin{IEEEproof}
According to Lemma \ref{lemma:LMI new DT-NI}, the NI property of the system (\ref{eq:G(z)}) implies that there exists a matrix $P=P^{\top}>0$ such that
\begin{equation}\label{eq:NI lemma}
	A^{\top}PA-P\leq 0 \quad \textnormal{and} \quad C = B^{\top}(I-A)^{-\top}P.
\end{equation}
We construct the following candidate Lyapunov function for the closed-loop interconnection of $G(z)$ and $\widehat {\mc H}$ shown in Fig.~\ref{fig:MIMO interconnection}:
\begin{align}
	W(x_k,\widetilde X_k) =& V(x_k)+\widehat V(\widetilde X_k)-x_k^{\top}C^{\top}\widetilde X_k\notag\\
	=& \frac{1}{2}x_k^{\top}Px_k+\frac{1}{2}\widetilde X_k^{\top}K^{-1}\widetilde X_k -x_k^{\top}C^{\top}\widetilde X_k. \label{eq:W}
\end{align}
Writing $W(x_k,\widetilde X_k)$ in matrix form, we have 
\begin{equation*}
	W(x_k,\widetilde X_k) = \frac{1}{2}\begin{bmatrix}
		x_k & \widetilde X_k
	\end{bmatrix}\begin{bmatrix}
		P & -C^{\top} \\-C & K^{-1}
	\end{bmatrix}\begin{bmatrix}
		x_k \\ \widetilde X_k
	\end{bmatrix}.
\end{equation*}
Let $\widehat P = \begin{bmatrix} 
		P & -C^{\top} \\-C & K^{-1}
	\end{bmatrix}$. Since $P>0$, the Schur complement theorem implies that $\widehat P>0$ if and only if
	\begin{equation}\label{eq:K condition}
		K^{-1}-CP^{-1}C^{\top}>0.
	\end{equation}
According to (\ref{eq:NI lemma}), we have $B = (I-A)P^{-1}C^{\top}$. Hence,
\begin{equation*}
	G(1) = C(I-A)^{-1}B = CP^{-1}C^{\top}.
\end{equation*}
The condition (\ref{eq:K condition}) can be rewritten as
\begin{equation*}
	K^{-1}-G(1)>0,
\end{equation*}
which is equivalent to the given condition $G(1)K<I$. Therefore, $\widehat P>0$ and the function $W(x_k,\widetilde X_k)$ is positive definite.

Now we take the difference between $W(x_{k+1},\widetilde X_{k+1})$ and $W(x_k,\widetilde X_k)$ to obtain
\begin{align}
	W(&x_{k+1},\widetilde X_{k+1})-W(x_k,\widetilde X_k)\notag\\
	 =&\ V(x_{k+1})+\widehat V(\widetilde X_{k+1})-x_{k+1}^{\top}C^{\top}\widetilde X_{k+1} - V(x_k)\notag\\
	 &-\widehat V(\widetilde X_k) + x_k^{\top}C\widetilde X_k\notag\\
	 \leq &\ u_k^{\top}(y_{k+1}-y_{k}) + E_k^{\top}(\widetilde X_{k+1}-\widetilde X_k)-x_{k+1}^{\top}C^{\top}\widetilde X_{k+1}\notag\\
	 &+ x_k^{\top}C\widetilde X_k\notag\\
	 =&\ \widetilde X_{k+1}^{\top}(y_{k+1}-y_{k})+y_k^{\top}(\widetilde X_{k+1}-\widetilde X_k)-y_{k+1}^{\top}\widetilde X_{k+1}\notag\\
	 &+y_k^{\top}\widetilde X_k\notag\\
	 =&\ 0,\label{eq:W difference}
\end{align}
where (\ref{eq:G(z) output eq}), (\ref{eq:bimodal multi-HIGS output}), (\ref{eq:setting u=tilde Y}) and (\ref{eq:setting E=y}) are also used.
This implies that the closed-loop system is Lyapunov stable. Moreover, the equality in (\ref{eq:W difference}) holds only if
\begin{align}
	V(x_{k+1}) - V(x_k)=&\ u_k^{\top}(y_{k+1}-y_{k}),\notag\\ 
\widehat V(\widetilde X_{k+1})-\widehat V(\widetilde X_k)=&\ E_k^{\top}(\widetilde X_{k+1}-\widetilde X_k).\label{eq:multi-HIGS lossless}
\end{align}
According to the proof of Theorem \ref{theorem:SANI property of multi-HIGS}, if (\ref{eq:multi-HIGS lossless}) holds, then for all $i \in \{1,2,\cdots, p\}$, we have
\begin{equation}\label{eq:HIGS i lossless}
\widetilde V_i\left(\widetilde x^{[i]}_{k+1}\right)-\widetilde V_i\left(\widetilde x^{[i]}_k\right)=e^{[i]}_k\left(\widetilde x^{[i]}_{k+1}-\widetilde x^{[i]}_k\right),
\end{equation}
with
\begin{equation}\label{eq:HIGS i storage}
\widetilde V_i\left(\widetilde x^{[i]}_k\right) = \frac{1}{2\kappa_i}\left(\widetilde x^{[i]}_k\right)^2.
\end{equation}
Substituting (\ref{eq:HIGS i storage}) in (\ref{eq:HIGS i lossless}), we obtain
\begin{equation}\label{eq:HIGS i lossless equation}
	\frac{1}{2\kappa_i}\left[\left(\widetilde x^{[i]}_{k+1}\right)^2-\left(\widetilde x^{[i]}_k\right)^2\right] = e^{[i]}_k\left(\widetilde x^{[i]}_{k+1}-\widetilde x^{[i]}_k\right).
\end{equation}
We discuss the conditions required for (\ref{eq:HIGS i lossless equation}) to hold in the cases of integrator mode and gain mode separately.

\medskip
\noindent
\textbf{Case 1 (integrator mode):} In this case, we have $\left(e^{[i]}_k,\widetilde x^{[i]}_{k|\mi} \right)\in \mc F_i$, and $\widetilde x^{[i]}_{k+1}$ is given by (\ref{eq:bi-HIGS i integrator mode}). Substituting (\ref{eq:bi-HIGS i integrator mode}) into (\ref{eq:HIGS i lossless equation}) yields
\begin{equation}\label{eq:integrator mode lossless equation}
	(\omega_i-2\kappa_i)\left(e^{[i]}_k\right)^2+2\widetilde x^{[i]}_ke^{[i]}_k=0.
\end{equation}
We now distinguish the cases $e^{[i]}_k \neq 0$ and $e^{[i]}_k = 0$. If $e^{[i]}_k \neq 0$, then (\ref{eq:integrator mode lossless equation}) implies
\begin{equation}\label{eq:tilde x in case1a}
	\widetilde x^{[i]}_k = \left(\kappa_i-\frac{\omega_i}{2}\right) e^{[i]}_k.
\end{equation}
Substituting (\ref{eq:tilde x in case1a}) in the inequality in (\ref{eq:F_i}) yields
\begin{equation*}
	\left(\kappa_i+\frac{\omega_i}{2}\right)\left(e^{[i]}_k\right)^2 \geq \frac{1}{\kappa_i}\left(\kappa_i+\frac{\omega_i}{2}\right)^2\left(e^{[i]}_k\right)^2,
\end{equation*}
which simplifies to $\omega_i\leq 0$ and hence contradicts the fact that $\omega_i>0$. Therefore, the case $e^{[i]}_k \neq 0$ cannot hold in the integrator mode. Thus, in integrator mode we must have $e^{[i]}_k = 0$. In this case, (\ref{eq:integrator mode lossless equation}) holds, and substituting $e^{[i]}_k = 0$ into \eqref{eq:F_i} yields $\widetilde x^{[i]}_k = 0$ and hence $\widetilde x^{[i]}_{k+1} = 0$. We show in the following that in this case, the HIGS channel $\mc H_i$ remains in integrator mode with zero state for all future times, regardless the values of future inputs. Since $\widetilde x^{[i]}_{k+1}=0$, then $\widetilde x^{[i]}_{k+1|\mi}=\omega_i \widetilde e^{[i]}_{k+1}$. According to (\ref{eq:F_i}), $\left(e^{[i]}_{k+1},\omega_i \widetilde e^{[i]}_{k+1}\right)\in \mc F_i$ if $\omega_i\left(e^{[i]}_{k+1}\right)^2\geq \frac{1}{\kappa_i}\omega_i^2\left(e^{[i]}_{k+1}\right)^2$. This holds for any $e_{k+1}^{[i]}$ because $0<\omega_i\leq \kappa_i$.

\medskip
\noindent
\textbf{Case 2 (gain mode):} In this case, we have $\left(e^{[i]}_k,\widetilde x^{[i]}_{k|\mi} \right)\notin \mc F_i$; that is,
\begin{equation}\label{eq:not in F_i}
	\left(\widetilde x^{[i]}_k+\omega_i e^{[i]}_k\right) e^{[i]}_k < \frac{1}{\kappa_i}\left(\widetilde x^{[i]}_k+\omega_i e^{[i]}_k\right)^2.
\end{equation}
Also, $\widetilde x^{[i]}_{k+1}$ is given by (\ref{eq:bi-HIGS i gain mode}). Substituting (\ref{eq:bi-HIGS i gain mode}) into (\ref{eq:HIGS i lossless equation}) yields
\begin{equation}\label{eq:tilde x case2}
	\widetilde x^{[i]}_k = \kappa_i e^{[i]}_k.
\end{equation}
Substituting (\ref{eq:tilde x case2}) into (\ref{eq:not in F_i}) shows that $e^{[i]}_k\neq 0$. Therefore, we have
\begin{equation*}
	\widetilde x^{[i]}_{k+1} = \kappa_i e^{[i]}_k = \widetilde x^{[i]}_k \neq 0.
\end{equation*}
From the analysis of the above two cases, we conclude that when $W(x_{k+1},\widetilde X_{k+1})-W(x_k,\widetilde X_k)=0$, we have the following conditions. For each HIGS $\mc H_i$: (i) if $e^{[i]}_k=0$, then it remains in the integrator mode and $\widetilde x^{[i]}_k=0$ for all future $k$, or (ii) if $e^{[i]}_k\neq 0$, it operates in the gain mode and $\widetilde x^{[i]}_k\neq 0$ are constant for all future $k$. According to (\ref{eq:bimodal multi-HIGS output}), $e^{[i]}_k$ also remains constant for all future $k$. This implies that it will remain in the gain mode for all future $k$.

Now we consider the case when $W(x_{k+1},\widetilde X_{k+1})-W(x_k,\widetilde X_k)=0$ for all $k\geq k_0$. In this case, each individual HIGS $\mc H_i$ will either remain in the integrator mode or remain in the gain mode. Also, according to the conditions in cases (i) and (ii), regardless of the active mode of each individual HIGS, its input and state will remain constant; i.e.,
\begin{align*}
	\widetilde X_k = \widetilde X_{k+1} = &\ \widetilde X_{k+2} = \cdots,\\
	E_k = E_{k+1} = &\ E_{k+2}=\cdots, \label{eq:E same}\\
	\widetilde X_k = &\ KE_k.
\end{align*}
Using the setting (\ref{eq:setting E=y}), we have
\begin{equation}\label{eq:y same}
	y_k = y_{k+1} = y_{k+2}=\cdots
\end{equation}
for all $k\geq k_0$. Since the system (\ref{eq:G(z)}) is minimal, it is observable. Therefore, (\ref{eq:y same}) implies that the state of the plant remains constant:
\begin{equation*}
	x_k = x_{k+1} = x_{k+2} = \cdots,\qquad \forall k\geq k_0.
\end{equation*}
Using $x_{k+1} = Ax_k+Bu_k$ with $u_k=\widetilde Y_k=\widetilde X_{k+1}=KE_{k+1}$ and $E_{k+1}=y_{k+1}=Cx_{k+1}=Cx_k$, we obtain, for all $k\ge k_0$,
\begin{equation}\label{eq:x_k E_k relation}
	x_k = Ax_k+BKE_k.
\end{equation}
Solving (\ref{eq:x_k E_k relation}), we get
\begin{equation*}
	x_k = (I-A)^{-1}BKE_k.
\end{equation*}
Hence,
\begin{equation}\label{eq:E = GKE}
	E_k = y_k = Cx_k = C(I-A)^{-1}BKE_k=G(1)KE_k.
\end{equation}
If $E_k\neq 0$ for some $k\geq k_0$, then \eqref{eq:E = GKE} implies $G(1)K=I$, contradicting the fact that $K^{-1}-G(1)>0$. Therefore, the only possible trajectory that remains lossless from some time onward satisfies $E_k=0$ and $y_k=0$ for all $k\geq k_0$, which by observability implies $x_k=0$ for all $k\geq k_0$.

Hence, except for the equilibrium, $W(x_{k+1},\widetilde X_{k+1})-W(x_k,\widetilde X_k)$ cannot remain zero and it will eventually decrease to zero, where we then have $x_k=0$ and $\widetilde X_k=0$. This proves asymptotic stability.
\end{IEEEproof}

\section{TRIMODAL DISCRETE-TIME HIGS AND NI CONTROL}\label{Trimodal discrete-time HIGS and NI control}
In this section, we investigate NI control using a discrete-time trimodal HIGS. First, we provide the system model for a single trimodal HIGS and demonstrate that it possesses the SANI property. Next, we introduce the trimodal multi-HIGS architecture and note that it retains the SANI property as a direct consequence of Theorem~\ref{theorem:SANI property of multi-HIGS}. Finally, we analyze the closed-loop stability for the interconnection of a MIMO ZOH-NI plant and a trimodal multi-HIGS controller.

\subsection{Trimodal single HIGS model and SANI property}
A trimodal version of discrete-time HIGS is introduced in \cite{sharif2024analysis}, which provides an alternative model to the bimodal version. The trimodal HIGS differs from the bimodal version by providing a refined way of projection into the sector of operation. Specifically, given an input $e_k$, if the integrator mode dynamics tend to provide an output outside the sector $[0,\kappa_h]$, the bimodal HIGS projects the output to the line $x_h=\kappa_he$. However, the sector $\mc F$ has two boundaries: $x_h = \kappa_h e$ and $x_h = 0$.

For a continuous-time HIGS, the dynamics of an integrator can only exit the sector $[0,\kappa_h]$ by crossing the boundary $x_h=\kappa_he$. Conversely, in the discrete-time case, a finite integration step can cause the dynamics to land in the second or fourth quadrants of the input-output plane, making the state closer to the zero boundary $x_h = 0$. To address this, the trimodal HIGS projects such states back to $x_h=0$ rather than $x_h=\kappa_h e$. We formulate the trimodal discrete-time HIGS model below. We adapt the indexing from the model in \cite{sharif2024analysis} to align with our nonlinear state-space description in \eqref{eq:DT_nonlinear}, rewriting the operating regions purely in terms of the current state and input:
\begin{equation}\label{eq:trimodal HIGS}
		\mathcal{H}:
		\begin{cases}
			x_h(k+1) = x_{\mi}(k), & \text{if } (e(k),x_{\mi}(k)) \in \mathcal{F},\\[0.5mm]
			x_h(k+1) = \kappa_he(k), & \text{if } (e(k),x_{\mi}(k)) \in \mathcal{F}_a,\\[0.5mm]
            x_h(k+1) = 0, & \text{if } (e(k),x_{\mi}(k)) \in \mc F_b,\\[0.5mm]
			y_h(k) = x_h(k+1),
		\end{cases}
\end{equation}
where $e(k),y_h(k),x_h(k)\in \mb R$ are the system input, output, and state, respectively, and $x_{\mi}(k)$ is the candidate integrator value given by the same formula as in \eqref{eq:x int}. The constant parameters $\omega_h\geq 0$ and $\kappa_h>0$ are the integrator frequency and the gain value, respectively. Here, $\mc F$ remains the same as in (\ref{eq:F}), and $\mc F_a$ and $\mc F_b$ are given by
\begin{align}
    \mc F_a =&\ \{(e(k),x_{\mi}(k))\in \mb R^2 \mid x_{\mi}(k)e(k)>\kappa_h e(k)^2\},\notag\\
   \mc F_b =&\ \Bigl\{(e(k),x_{\mi}(k))\in \mb R^2 \Bigm| \notag\\
   \quad & x_{\mi}(k)e(k)<0 \ \vee\ \bigl(e(k)=0 \wedge x_{\mi}(k)\neq 0\bigr)\Bigr\}.\label{eq:Fb}
\end{align}
To streamline the notation in the subsequent analysis, we continue to denote $e(k)$, $x_h(k)$, $x_{\mi}(k)$ and $y_h(k)$ by $e_k$, $\widetilde x_k$, $\widetilde x_{k|\mi}$ and $\widetilde y_k$, respectively.

\begin{remark}
The trimodal HIGS model in \eqref{eq:trimodal HIGS} represents an equivalent state-space realization of the controller presented in \cite{sharif2024analysis}. While the original model expresses the switching conditions using both $e(k)$ and $e(k-1)$, the representation formulated here defines the operating regions entirely in terms of the current state and input, eliminating the explicit dependence on $e(k-1)$. The three defined regions $\mathcal F$, $\mathcal F_a$, and $\mathcal F_b$ correspond directly to the three projection operations described in \cite{sharif2024analysis}: pure integration, projection onto the gain boundary $x_h=\kappa_h e$, and projection onto the zero boundary $x_h=0$.
\end{remark}

\begin{theorem}\label{theorem:single trimodal HIGS SANI}
    The trimodal HIGS of the form (\ref{eq:trimodal HIGS}) is an SANI system with the storage function (\ref{eq:HIGS storage function}) satisfying (\ref{eq:NNI ineq for HIGS aux}).
\end{theorem}
\begin{IEEEproof}
When $(e(k),x_{\mi}(k))\in \mc F \cup \mc F_a$, the update law for $x_h(k+1)$ is identical to that of the bimodal HIGS. Consequently, the SANI property for these regions follows directly from Theorem~\ref{theorem:single HIGS SANI}. 

For the zeroing mode where $(e(k),x_{\mi}(k))\in \mc F_b$, the state update is $x_h(k+1)=0$. In this case, the required SANI inequality \eqref{eq:NNI ineq for HIGS aux} simplifies to:
\begin{equation}\label{eq:trimodal SANI Fb NI condition}
        \frac{1}{2\kappa_h}\widetilde x_k^2\geq e_k\widetilde x_k.
\end{equation}
If $e_k=0$, \eqref{eq:trimodal SANI Fb NI condition} is always satisfied. If $e_k\neq 0$, dividing by $e_k^2$ allows \eqref{eq:trimodal SANI Fb NI condition} to be rewritten as a quadratic inequality:
\begin{equation*}
\left(\frac{\widetilde x_k}{e_k}\right)^2 - 2\kappa_h\left(\frac{\widetilde x_k}{e_k}\right)\ge 0,
\end{equation*}
which holds when $\frac{\widetilde x_k}{e_k}\leq 0$ or $\frac{\widetilde x_k}{e_k}\geq 2\kappa_h$. The condition $x_{\mi}(k)e_k<0$ in $\mc F_b$ given by (\ref{eq:Fb}) implies that
\begin{equation*}
\widetilde x_k e_k+\omega_h e_k^2<0,
\end{equation*}
which, in the case $e_k\neq 0$, implies $\frac{\widetilde x_k}{e_k}+\omega_h<0$. Hence, $\frac{\widetilde x_k}{e_k}<0$, and (\ref{eq:trimodal SANI Fb NI condition}) is satisfied. Therefore, the trimodal HIGS is an SANI system.
\end{IEEEproof}

\subsection{Trimodal multi-HIGS model and SANI property}
A trimodal discrete-time multi-HIGS is constructed by connecting $p$ trimodal HIGS of the form (\ref{eq:trimodal HIGS}) in parallel, as shown in Fig.~\ref{fig:multi_HIGS}. The stacked input, output and state of the multi-HIGS $\widehat{\mc H}$ are still given by (\ref{eq:multi_HIGS_input})--(\ref{eq:multi_HIGS_state}). 

Assume that the $i$-th HIGS $\mc H_i$ has an integrator frequency $\omega_i\geq 0$ and a gain value $\kappa_i>0$, with scalar input $e^{[i]}_k$, output $\widetilde y^{[i]}_k$ and state $\widetilde x^{[i]}_k$. The associated candidate integrator value for this channel is:
\begin{equation}\label{eq:trimodal multi-HIGS int}
\widetilde x^{[i]}_{k|\mi} = \widetilde x^{[i]}_k + \omega_i e^{[i]}_k.
\end{equation}
The operating sector regions for the $i$-th channel are defined as:
\begin{align*}
\mc F_i =&\ \{(e,x_{\mi})\in\mb R^2 \mid x_{\mi}e\ge\tfrac{1}{\kappa_i}x_{\mi}^2\},\\
\mc F_{a,i} =&\ \{(e,x_{\mi})\in\mb R^2 \mid x_{\mi}e>\kappa_i e^2\},\\
\mc F_{b,i} =&\ \{(e,x_{\mi})\in\mb R^2 \mid x_{\mi}e<0 \vee (e=0\wedge x_{\mi}\neq 0)\}.
\end{align*}
The state-space representation of the trimodal multi-HIGS can then be written as:
\begin{subequations}\label{eq:trimodal multi-HIGS}
	\begin{align}
e^{[i]}_k =&\ \theta_i^{\top} E_{k},\\
\widetilde x^{[i]}_{k+1}
  =&\ \widetilde x^{[i]}_{k|\mi},  &&\hspace{-1cm}\text{if } (e^{[i]}_k,\widetilde x^{[i]}_{k|\mi}) \in \mathcal F_i, \\
\widetilde x^{[i]}_{k+1}
  =&\ \kappa_i e^{[i]}_k,         &&\hspace{-1cm} \text{if } (e^{[i]}_k,\widetilde x^{[i]}_{k|\mi}) \in \mathcal F_{a,i}, \\
\widetilde x^{[i]}_{k+1}
  =&\ 0,  && \hspace{-1cm}\text{if } (e^{[i]}_k,\widetilde x^{[i]}_{k|\mi}) \in \mathcal F_{b,i},\\
\widetilde X_k =&\ \left[
	\widetilde x^{[1]}_k, \widetilde x^{[2]}_k,\cdots, \widetilde x^{[p]}_k
			\right]^\top ,\\
\widetilde Y_k =&\ \widetilde X_{k+1}.\label{eq:trimodal multi-HIGS output}
	\end{align}
\end{subequations}

By Theorem~\ref{theorem:single trimodal HIGS SANI}, each trimodal HIGS $\mc H_i$ has the storage function
\begin{equation}\label{eq:HIGS i storage trimodal}
\widetilde V_i\left(\widetilde x^{[i]}_k\right) = \frac{1}{2\kappa_i}\left(\widetilde x^{[i]}_k\right)^2
\end{equation}
and satisfies \eqref{eq:SANI ineq for HIGS i generic}. Thus, the trimodal multi-HIGS \eqref{eq:trimodal multi-HIGS} is also SANI with the storage \eqref{eq:multi-HIGS storage} by Theorem~\ref{theorem:SANI property of multi-HIGS}.
\begin{corollary}\label{corollary:trimodal SANI}
The trimodal multi-HIGS given by \eqref{eq:trimodal multi-HIGS} is an SANI system from the input $E_k$ to the output $\widetilde Y_k$, with the storage function \eqref{eq:multi-HIGS storage}.
\end{corollary}
\begin{IEEEproof}
This result follows directly from Theorems~\ref{theorem:single trimodal HIGS SANI} and \ref{theorem:SANI property of multi-HIGS}.
\end{IEEEproof}

\subsection{Control of NI systems using trimodal multi-HIGS}
We now show that a linear ZOH-NI plant can also be stabilized by a trimodal discrete-time multi-HIGS in positive feedback. We consider the same interconnection as in Fig.~\ref{fig:MIMO interconnection} with the settings \eqref{eq:setting u=tilde Y} and \eqref{eq:setting E=y}.
\begin{theorem}
\label{theorem:trimodal stability}
  Consider a system with transfer matrix $G(z)$ and a minimal realization \eqref{eq:G(z)} where $\det(I-A)\neq 0$. Suppose $G(z)$ is ZOH-NI with a positive definite quadratic storage function. Also, consider a trimodal multi-HIGS $\widehat {\mc H}$ of the form (\ref{eq:trimodal multi-HIGS}) such that $0<\omega_i\leq \kappa_i$ for all $i\in \{1,2,\cdots,p\}$ and $K^{-1} - G(1) > 0$ where $K$ is given by (\ref{eq:K_def}). Then the closed-loop interconnection of $G(z)$ and the trimodal multi-HIGS $\widehat{\mc H}$ is asymptotically stable.
\end{theorem}

\begin{IEEEproof}
We use the same candidate Lyapunov function $W$ in \eqref{eq:W} as in the bimodal case. Its positive definiteness follows exactly as in the proof of Theorem~\ref{theorem:bimodal stability}.

As shown in Corollary \ref{corollary:trimodal SANI}, the trimodal multi-HIGS is SANI with storage function \eqref{eq:multi-HIGS storage}. Conducting the same dissipation calculation as given by \eqref{eq:W difference} in the proof of Theorem~\ref{theorem:bimodal stability} yields
\begin{equation*}
  W(x_{k+1},\widetilde X_{k+1}) - W(x_k,\widetilde X_k) \le 0
\end{equation*}
for all $k$, with equality if and only if both subsystems are lossless at step $k$; that is:
\begin{align}
  V(x_{k+1}) - V(x_k)
  &= u_k^{\top} (y_{k+1}-y_k),\notag\\
  \widehat V(\widetilde X_{k+1}) - \widehat V(\widetilde X_k)
  &= E_k^{\top} (\widetilde X_{k+1}-\widetilde X_k). \label{eq:multi lossless trimodal}
\end{align}
Similar to the analysis in the proof of Theorem \ref{theorem:bimodal stability}, \eqref{eq:multi lossless trimodal} implies that each individual HIGS
$\mc H_i$ is lossless, i.e.,
\begin{equation}\label{eq:HIGS i lossless trimodal}
  \widetilde V_i(\widetilde x^{[i]}_{k+1})
  - \widetilde V_i(\widetilde x^{[i]}_k)
  = e^{[i]}_k\bigl(\widetilde x^{[i]}_{k+1}-\widetilde x^{[i]}_k\bigr),
\end{equation}
with $\widetilde V_i$ given by \eqref{eq:HIGS i storage trimodal}. Substituting \eqref{eq:HIGS i storage trimodal} into
\eqref{eq:HIGS i lossless trimodal} yields
\begin{equation}\label{eq:HIGS i lossless eq trimodal}
  \frac{1}{2\kappa_i}\Bigl[
    (\widetilde x^{[i]}_{k+1})^2 - (\widetilde x^{[i]}_k)^2
  \Bigr]
  = e^{[i]}_k\bigl(\widetilde x^{[i]}_{k+1}-\widetilde x^{[i]}_k\bigr).
\end{equation}
We now analyze \eqref{eq:HIGS i lossless eq trimodal} for each mode of the
trimodal HIGS $\mc H_i$. The integrator mode case $\bigl(e^{[i]}_k,\widetilde x^{[i]}_{k|\mi}\bigr)\in\mc F_i$ and the gain mode case $\bigl(e^{[i]}_k,\widetilde x^{[i]}_{k|\mi}\bigr)\in\mc F_{i,a}$ follow the same analysis as in the proof of Theorem \ref{theorem:bimodal stability}. Now, we consider the zeroing mode case where $\bigl(e^{[i]}_k,\widetilde x^{[i]}_{k|\mi}\bigr)\in\mc F_{i,b}$.

\medskip
\noindent
\textbf{Zeroing mode Case:}
$\bigl(e^{[i]}_k,\widetilde x^{[i]}_{k|\mi}\bigr)\in\mc F_{b,i}$, and
\[
  \widetilde x^{[i]}_{k+1} = 0.
\]
In this case, \eqref{eq:HIGS i lossless eq trimodal} becomes
\begin{equation*}
  -\frac{1}{2\kappa_i}(\widetilde x^{[i]}_k)^2
  = e^{[i]}_k(0-\widetilde x^{[i]}_k)
  = -e^{[i]}_k \widetilde x^{[i]}_k,
\end{equation*}
or equivalently
\begin{equation}\label{eq:zero mode equality simplified}
  \frac{1}{2\kappa_i}(\widetilde x^{[i]}_k)^2
  = e^{[i]}_k \widetilde x^{[i]}_k.
\end{equation}
If $e^{[i]}_k=0$, then $\mc F_{b,i}$ requires $\widetilde x_{k|\mi}^{[i]}\neq 0$. This implies $\widetilde x^{[i]}_k\neq 0$ by \eqref{eq:trimodal multi-HIGS int}, which contradicts \eqref{eq:zero mode equality simplified}. If $e^{[i]}_k\neq 0$, then $\mc F_{b,i}$ implies that $\widetilde x^{[i]}_{k|\mi}e^{[i]}_k<0$. Expanding this using \eqref{eq:trimodal multi-HIGS int} yields $e^{[i]}_k \widetilde x^{[i]}_k+\omega_i \l(e^{[i]}_k\r)^2<0$. Because $\omega_i > 0$, this strict inequality necessitates that $e^{[i]}_k \widetilde x^{[i]}_k<0$, which contradicts with \eqref{eq:zero mode equality simplified}. Hence, when $W(x_{k+1},\widetilde X_{k+1}) - W(x_k,\widetilde X_k)=0$, no HIGS channel $\mc H_i$ can operate in the zeroing mode. In other words, whenever the closed-loop system is lossless at step $k$, each HIGS
$\mc H_i$ must be either in the integrator mode or in the gain
mode.

The remainder of the proof follows identically from the final steps of the proof of Theorem~\ref{theorem:bimodal stability}.
\end{IEEEproof}

\begin{figure}[t]
    \centering
    \begin{overpic}[width=1\linewidth]{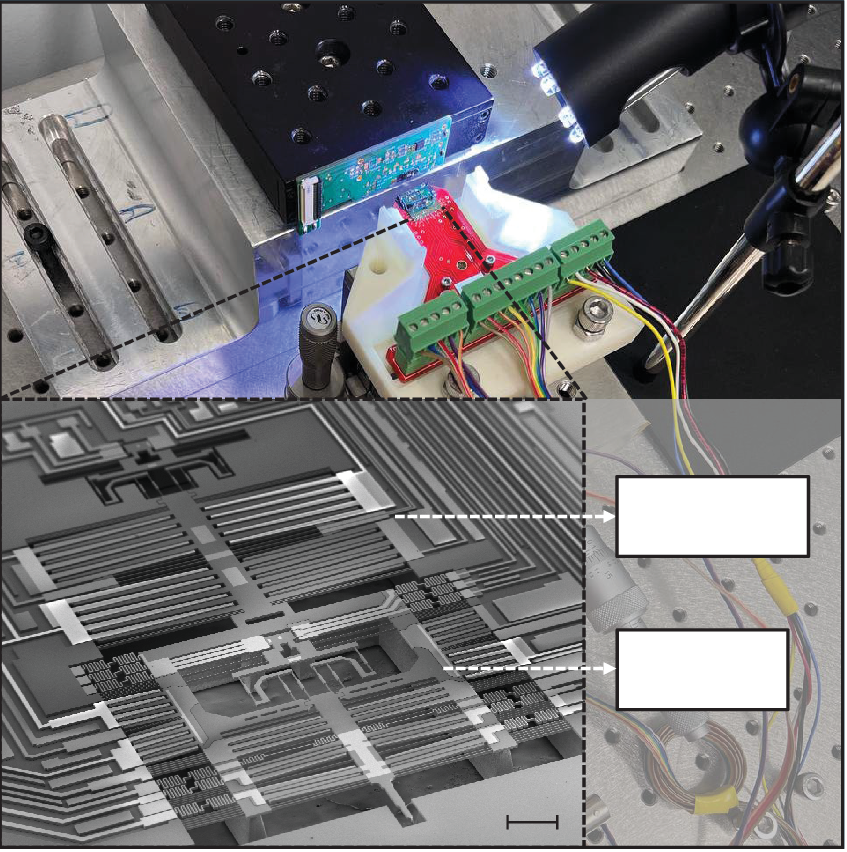}
        \put(73.7,20.6){%
            \parbox[c]{3cm}{
                \fontsize{8}{9}\selectfont 
                Inner-Stage\\
                Force Sensor
            }
        }
        \put(73.7,38.2){%
            \parbox[c]{3cm}{
                \fontsize{8}{9}\selectfont 
                Outer-Stage\\
                Nanopositioner
            }
        }
        \put(57,5){%
            \parbox[c]{3cm}{
                \fontsize{8}{9}\selectfont 
                $500\mu m$
            }
        }
    \end{overpic}
    \caption{\small The testbed with the dual-stage MEMS force sensor mounted on a custom-made PCB with a scanning electron microscopy (SEM) image of the MEMS device reported in \cite{dadkhah2024design}.}
    \label{fig:SEM}
\end{figure}

\section{APPLICATION: A MEMS FORCE SENSOR}\label{Application: A MEMS Force Sensor}

\subsection{Dual-stage MEMS force sensor}
We develop a MIMO HIGS controller for a dual-axis MEMS force sensor featuring collocated electrostatic actuators and electrothermal displacement sensors that generate and detect bidirectional motion along the X-axis. A scanning electron microscope (SEM) image of the device is presented in Fig.~\ref{fig:SEM}. The double-stage force sensor integrates nanopositioning and force sensing systems within a total length of $7.85~\mathrm{mm}$. The force sensing stage, measuring $3.86~\mathrm{mm}$ in length and $2.82~\mathrm{mm}$ in width, incorporates a shuttle beam with dimensions of $2.37~\mathrm{mm} \times 0.250~\mathrm{mm}$, ending in a sharp $2 \times 25~\mu \mathrm{m}$ tip. The nanopositioning stage provides a maximum linear displacement of $5.3~\mu \mathrm{m}$, while the force sensing stage operates within a linear travel range of $2.29~\mu \mathrm{m}$ and an open-loop force range of $86.28~\mu \mathrm{N}$. The design and characterization of the device are detailed in \cite{dadkhah2024design}.

\begin{figure}
    \centering
    \includegraphics[width=0.9\linewidth]{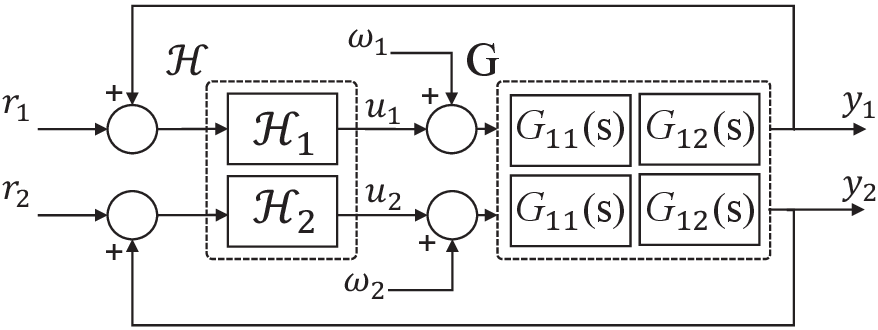}
    \caption{\small Closed-loop interconnection of a TITO multi-HIGS and the TITO dual-stage MEMS force sensor G(s) for damping.}
    \label{fig:DampedScheme}
\end{figure}

\begin{figure*}
\centering
\includegraphics[width=1\textwidth]{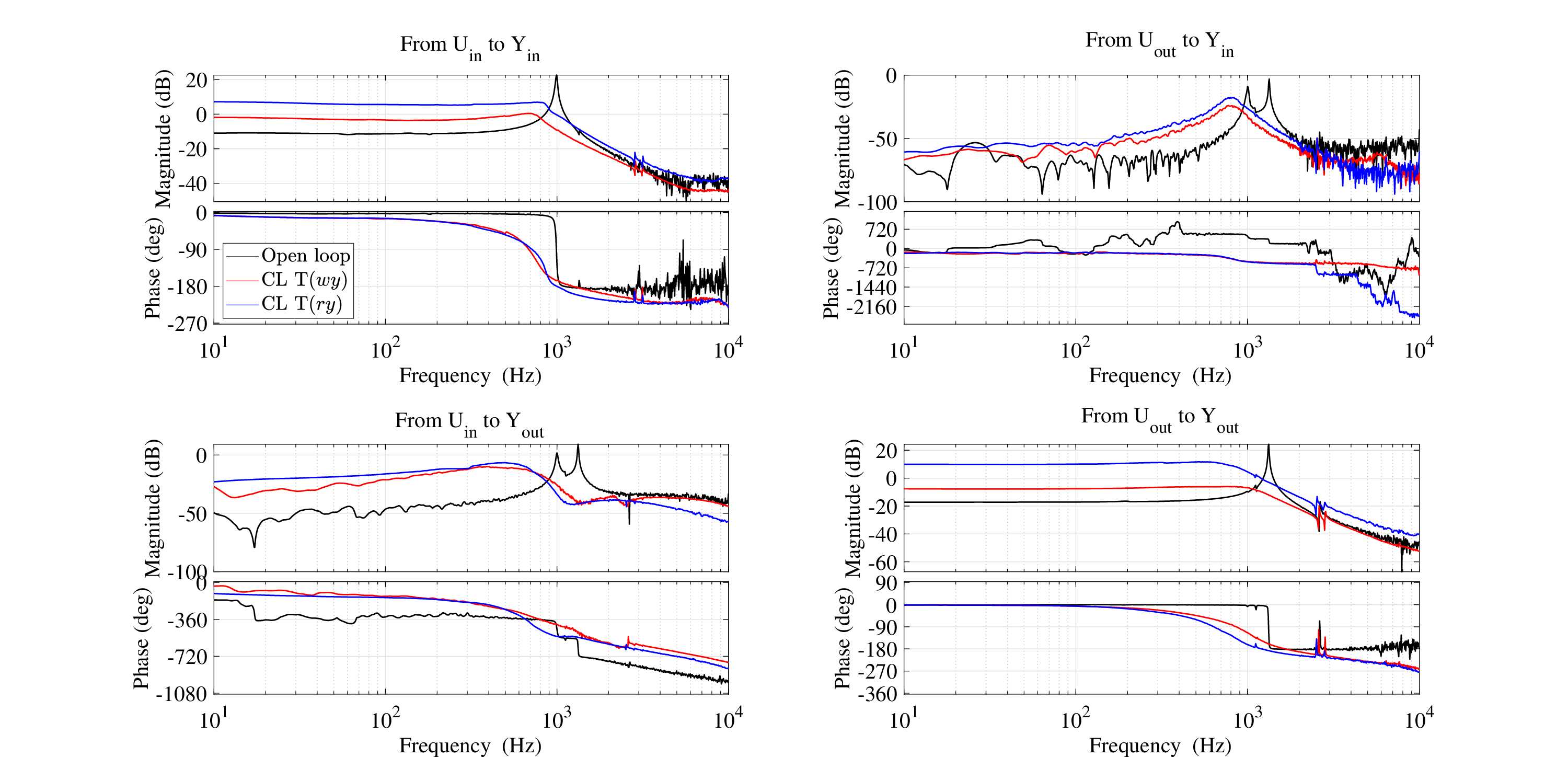}
\caption{\small Frequency responses of the MEMS force sensor in open and closed loop with TITO bimodal multi-HIGS in positive feedback.}
\label{fig:OLCL}
\end{figure*}

The frequency-domain response of the dual-stage MEMS force sensor is obtained using an Onosokki FFT analyzer (CF-9400). The device is modeled as a two-input two-output (TITO) system, where the inputs correspond to the voltages applied to the electrostatic actuators via a $20\times$ voltage amplifier, and the outputs are the voltage signals from the inner-stage and outer-stage sensors. The frequency response function (FRF) of the TITO MEMS force sensor is defined as:
\begin{flalign*}
&G(j\omega)=
\begin{bmatrix}
     G_{11}(j\omega)&G_{12}(j\omega)\\
G_{21}(j\omega)&G_{22}(j\omega)
 \end{bmatrix}=
 \begin{bmatrix}
     \frac{Y_{in}}{U_{in}}(j\omega) &  \frac{Y_{in}}{U_{out}}(j\omega)\\
\frac{Y_{out}}{U_{in}}(j\omega) &  \frac{Y_{out}}{U_{out}}(j\omega)
 \end{bmatrix}&
\end{flalign*}
where $U_{in,out}$ and $Y_{in,out}$ represent the Fourier transforms of inner-stage force sensor and outer-stage nanopositioner actuation voltages and sensor outputs, respectively. A sinusoidal chirp signal sweeping across a $10~\mathrm{kHz}$ bandwidth is applied to the actuators, and the corresponding responses are captured from the electrothermal sensors' trans-impedance readout circuit. Fig.~\ref{fig:OLCL} shows the obtained open-loop frequency response of the TITO system. The dominant resonant modes of the inner-stage force sensor and outer-stage nanopositioner occur at $993~\mathrm{Hz}$ and $1326~\mathrm{Hz}$, respectively. Due to the mechanical coupling between the stages, the coupling terms also appear at these frequencies. A fourth-order state-space model is estimated using the plant frequency response data and the MATLAB $n4sid$ system identification toolbox: 
\begin{flalign*}&A=
{\small\begin{bmatrix}
     -104.10&5622&3693&-1174\\
    -5238&-17.10&596.60&6212\\
    -2765&-461.60&-37.94&-7586\\
    646.40&-3833&5758&-22.39\\
 \end{bmatrix}},&
\end{flalign*}

\begin{flalign*}&
B={\small\begin{bmatrix}
     23.58&3.81\\
     9.50&-4.25\\
     2.36&4.89\\
    5.69&-14.11\\
 \end{bmatrix}},\ C={\small\begin{bmatrix}
     18.69&6.88\\
     -45.25&-43.13\\
     -50.40&40.23\\
     0.23&26.27\\
 \end{bmatrix}^{\top}},\ D=0.&
\end{flalign*}

To assess the NI property of the TITO MEMS force sensor, we evaluate the eigenvalues of the matrix $j(G(j\omega)-G^*(j\omega))$ (see e.g., \cite{lanzon2008stability}), showing that the system retains the NI property up to $1004~\mathrm{Hz}$. As shown in Fig.~\ref{fig:OLCL}, the phase response remains within the $0$ to $-180^{\circ}$ range up to this frequency but gradually decreases beyond that point, indicating loss of the NI property. Therefore, within this bandwidth, the system is NI, indicating that it is ZOH-NI under sample-and-hold process in digital control. Thus, the internal stability of the closed-loop system is ensured when a multi-HIGS controller with suitable parameters is applied in positive feedback. Given that the operating frequency range of the MEMS force sensor lies within this range, the stability of the positively fed-back closed-loop system with the multi-HIGS controller is inherently assured.

\subsection{Discrete-time multi-HIGS controller design}
According to Theorems \ref{theorem:bimodal stability} and \ref{theorem:trimodal stability}, discrete-time multi-HIGS of the forms (\ref{eq:bimodal multi-HIGS}) and (\ref{eq:trimodal multi-HIGS}) can guarantee closed-loop stability when interconnected to the dual-stage MEMS force sensor integrator frequencies $\omega_{1}$ and $\omega_{2}$ and gain values $\kappa_{1}$ and $\kappa_{2}$ are such that $0 < \omega_i\leq \kappa_i$ and
\begin{equation}\label{eq:example DC gain condition}
	K^{-1}-G(1)>0, 
\end{equation}
where $K = \mathrm{diag}\{\kappa_1,\kappa_2\}$ and $G(1)$ is the DC-gain of the dual-stage MEMS force sensor. Since the condition \eqref{eq:example DC gain condition} is a constraint on $\kappa_i$ $(i=1,2)$ only, we can tune $\omega_i$ after the values of $\kappa_i$ are determined in order to achieve desired closed-loop performance.

Although the controller is designed and implemented directly in discrete time, we refer to the continuous-time HIGS describing function to aid the initial parameter tuning, since it provides an approximate frequency-domain characterization of the HIGS performance. Specifically, the describing function of a HIGS provides a quasi-linear frequency-domain approximation relating a sinusoidal input of frequency $\omega$ to the fundamental harmonic of the corresponding HIGS output \cite{heertjes2019hybrid}:
\begin{equation*}
    \begin{split}
        D_\mathcal{H}(j\omega) &= \frac{\omega_h}{j\omega}[\frac{\gamma(\omega)}{\pi} + j\frac{e^{-2j\gamma(\omega)}-1}{2\pi}-4j\frac{e^{-j\gamma(\omega)}-1}{2\pi}]\\
        &+k_h[\frac{\pi-\gamma(\omega)}{\pi} + j\frac{e^{-2j\gamma(\omega)}-1}{2\pi}],
    \end{split}
\end{equation*}
with $\gamma(\omega) = 2\,\tan^{-1}(\frac{k_h\omega}{\omega_h})$. As noted earlier, the discrete-time gain parameter $\kappa_i$ plays the same sector-defining role as the continuous-time gain $k_h$, while the discrete-time integrator parameter $\omega_i$ corresponds to the sampled integrator increment and is related to the continuous-time integrator rate through $\omega_i \approx \omega_h T_s$, where $T_s=1/f_s$ is the sampling period. In the experiments, the sampling frequency is $f_s=50~\mathrm{kHz}$, which gives $T_s=20~\mu\mathrm{s}$. This correspondence is used here only to guide parameter selection of each diagonal HIGS channel while the controller itself remains the discrete-time multi-HIGS defined in \eqref{eq:bimodal multi-HIGS} and \eqref{eq:trimodal multi-HIGS}.

Using this channel-wise approximation as an initial tuning guide, we selected
\begin{align*}
    K=&\ \mathrm{diag}\{0.75\,G_{11}(1)^{-1},\,0.85\,G_{22}(1)^{-1}\},\\
    \Omega=&\ \mathrm{diag}\{1.4\,\omega_{n1}T_s,\;3.2\,\omega_{n2}T_s\},
\end{align*}
where $G_{ii}(1)$ denotes the DC-gain and $\omega_{ni}$ denotes the dominant natural frequency of the force sensing and nanopositioning stages, respectively. The resulting discrete-time HIGS parameters are
\begin{equation*}
K=
\begin{bmatrix}
     2.81&0\\
     0&6.25
 \end{bmatrix},\qquad
\Omega=
\begin{bmatrix}
   0.174&0\\
   0&0.532
 \end{bmatrix}.
\end{equation*}
These values were then implemented directly in the discrete-time bimodal and trimodal multi-HIGS controllers and validated experimentally.

\subsection{Closed-loop performance}
The closed-loop control scheme of the TITO dual-stage MEMS force sensor with a TITO multi-HIGS is shown in Fig.~\ref{fig:DampedScheme}. In this setup, $r$ represents the reference signal and $w$ is the input disturbance to the plant. Since the main goal in damping the resonant modes is to improve input disturbance rejection rather than to emphasize reference tracking, the analysis centers on the input disturbance to the sensor output. The closed-loop disturbance rejection ($T_{wy}$) frequency response of the MEMS force sensor in positive feedback with the TITO bimodal multi-HIGS is experimentally obtained and shown in Fig.~\ref{fig:OLCL}. As seen, the HIGS element effectively damps the resonant modes of both stages, thereby eliminating cross-coupling between them, as evident from the closed-loop frequency responses of $T_{12}$ and $T_{21}$ in Fig.~\ref{fig:OLCL}. The HIGS parameters can also be adjusted to achieve a zero closed-loop DC-gain for improved reference tracking. However, this comes at the expense of reduced damping performance.

The closed-loop damping performance was further investigated through time-domain experiments. For real-time measurements, the control scheme depicted in Fig.~\ref{fig:DampedScheme} was digitally implemented using a dSPACE rapid prototyping system with a sampling rate of $50~\mathrm{kHz}$. Step inputs with an amplitude of $0.2~\mathrm{V}$ were applied to both open-loop and closed-loop systems, and the corresponding sensor outputs were recorded in real-time. The normalized results confirm the effectiveness of the multi-HIGS approach in stabilizing the MEMS force sensor within the expected stability conditions, which validates the theoretical expectations as long as the system is NI. 

As depicted in Fig.~\ref{fig:OL_Damped}, the TITO multi-HIGS in positive feedback enhances system performance by effectively damping its resonant modes. This leads to reduced overshoot and settling time for both the inner-stage force sensor and the outer-stage nanopositioner. However, the closed-loop configuration changes the system's low-frequency gain by introducing an additional gain, indicating that the chosen parameters contribute to an increase in the system's DC-gain.

\subsection{Discrete-time bimodal and trimodal HIGS }
\begin{figure}
\centering
\includegraphics[width=0.5\textwidth]{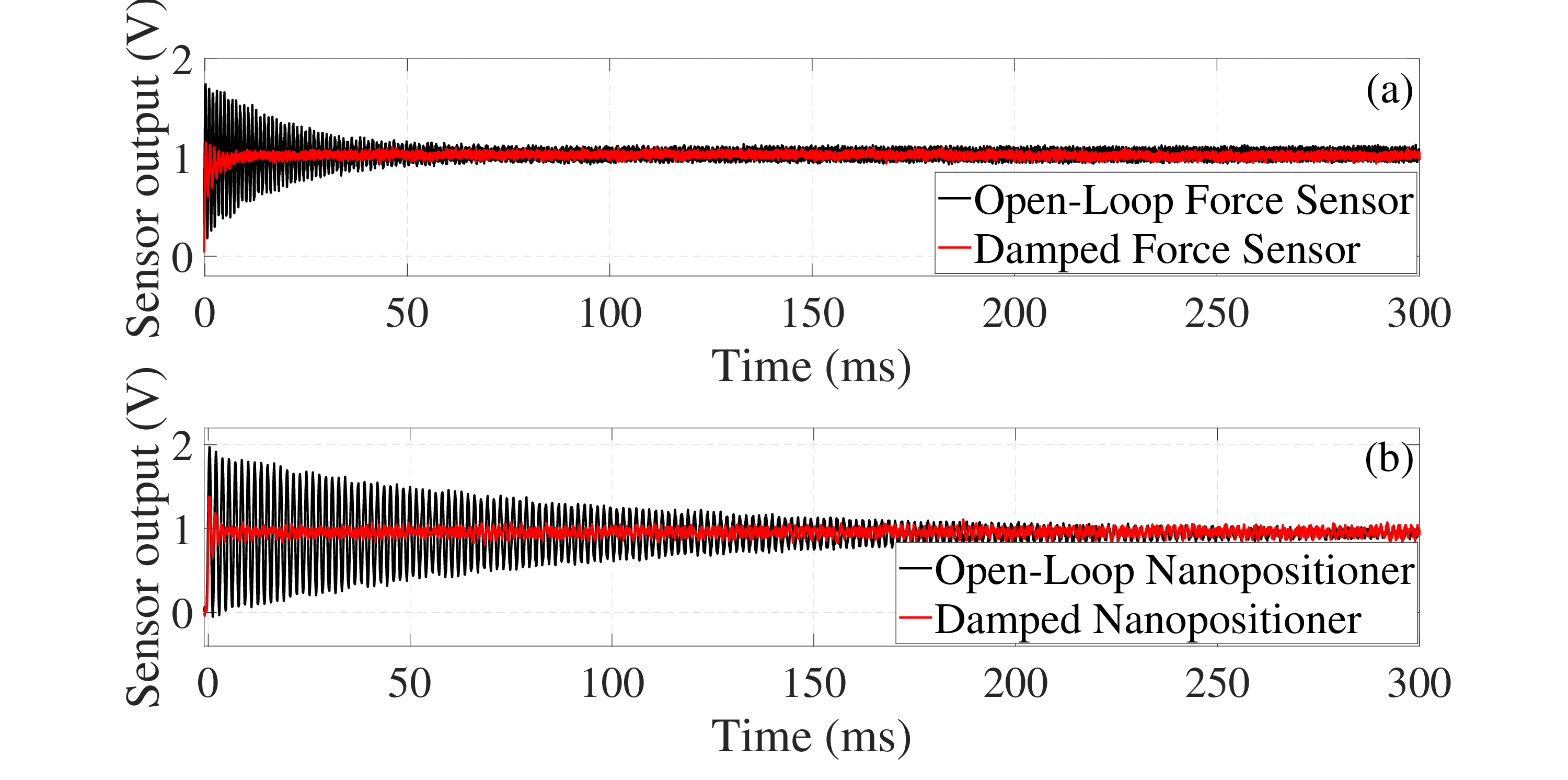}
\caption{\small Normalized step responses of the a) inner-stage force sensor and the b) outer-stage nanopositioner in
open-loop and closed-loop with TITO bimodal multi-HIGS.
}
\label{fig:OL_Damped}
\end{figure}

To compare the stability and performance of the bimodal and trimodal discrete-time HIGS controllers in a positive-feedback interconnection with the MEMS force sensor, both controllers were implemented digitally in real time and evaluated through closed-loop time-domain experiments. A step input of $0.2~\mathrm{V}$ was applied as a disturbance. Fig.~\ref{fig:BMvsTM} shows the normalized sensor outputs and control signals for each controller. The results indicate that the controllers achieve comparable damping and steady-state response. We observe that the closed-loop systems in both cases are asymptotically stable. The trimodal HIGS, however, demonstrates improved control near the equilibrium point. In this region, the sector constraint $\mc F_b$ in (\ref{eq:trimodal HIGS}) helps avoids large jumps in the control signal, which reduces the noise level relative to the bimodal controller. The 1‑$\sigma$ resolution of the control signal is $3~\mathrm{mV}$ for the bimodal HIGS and is reduced to approximately $0.7~\mathrm{mV}$ for the trimodal HIGS near the equilibrium point. Away from equilibrium, both controllers operate with similar performance, leading to comparable control effort and noise levels. The closed-loop frequency responses of the force sensor using TITO trimodal multi-HIGS are identical to the bimodal case presented in Fig.~\ref{fig:OLCL}, and are therefore omitted here to avoid cluttering the figure.

\section{CONCLUSION}\label{Conclusion}
This paper investigates the application of discrete-time bimodal and trimodal multi-HIGS controllers to a TITO NI MEMS force sensor equipped with collocated electrostatic actuators and electrothermal position sensors. Theoretical analysis, supported by experimental validation, demonstrates that the positive-feedback interconnection of a discrete-time multi-HIGS structure with suitable parameters and a MIMO NI system is asymptotically stable. To this end, discrete-time TITO HIGS elements are systematically tuned and digitally implemented. Frequency- and time-domain analyses verify the stability of the closed-loop system and substantiate the effectiveness of the multi-HIGS element as an NI controller. The experimental findings also indicate that the trimodal HIGS provides more precise control compared to the bimodal case. The damping introduced by the HIGS elements attenuate inter-stage cross-coupling leading to decoupling the stage dynamics. As a result, the damped system can be approximated as two cascaded SISO subsystems, which simplifies the subsequent tracking controller design. The development of reference tracking framework is deferred to future work.
\begin{figure}
\centering
\includegraphics[width=0.5\textwidth]{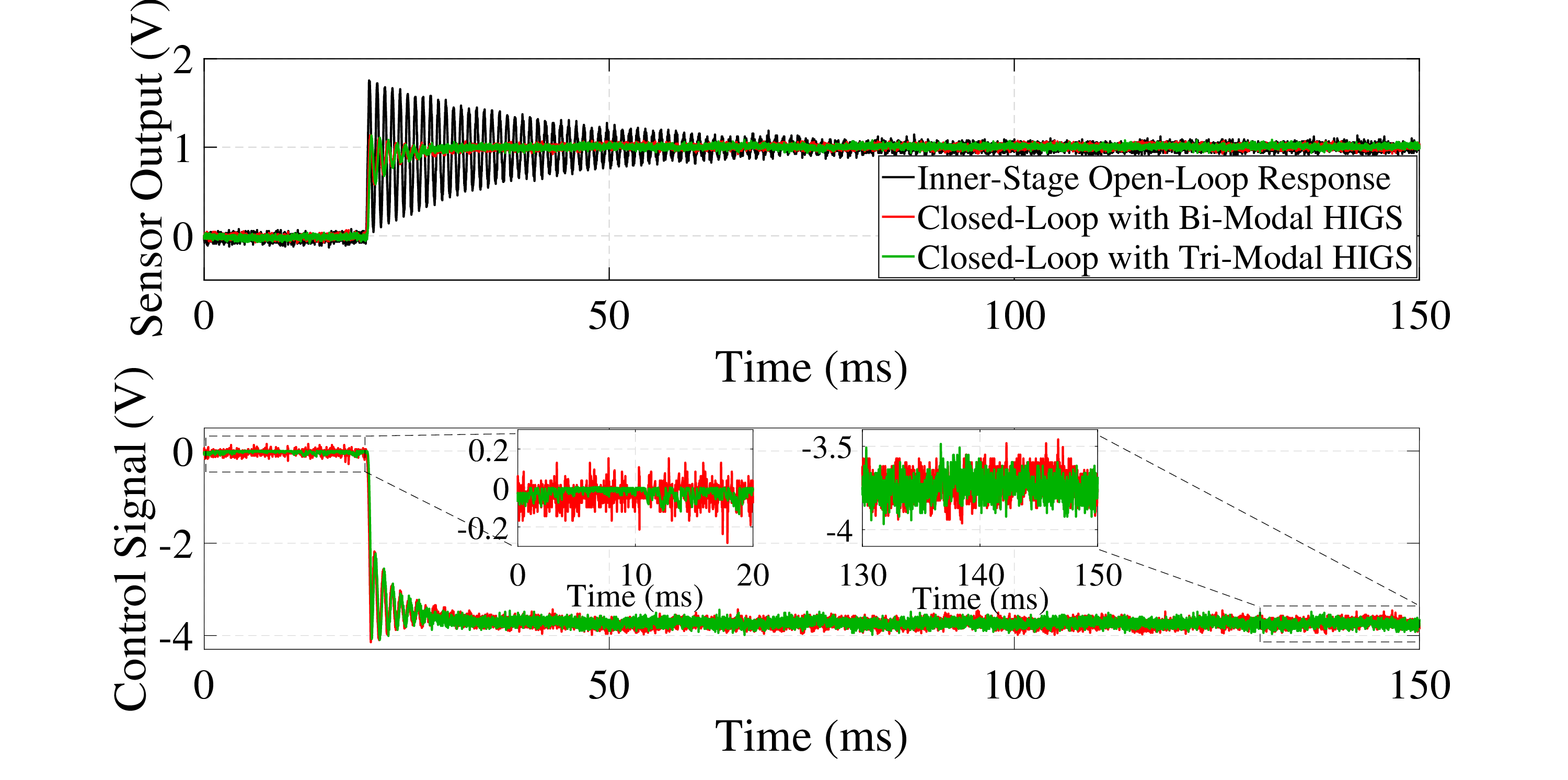}
\caption{\small Closed-loop performance comparison of bimodal and trimodal discrete-time HIGS in positive feedback interconnection with the force sensing system.
}
\label{fig:BMvsTM}
\end{figure}
\bibliographystyle{IEEEtran}
\bibliography{reference.bib}

\begin{thebibliography}{10}
\providecommand{\url}[1]{#1}
\csname url@samestyle\endcsname
\providecommand{\newblock}{\relax}
\providecommand{\bibinfo}[2]{#2}
\providecommand{\BIBentrySTDinterwordspacing}{\spaceskip=0pt\relax}
\providecommand{\BIBentryALTinterwordstretchfactor}{4}
\providecommand{\BIBentryALTinterwordspacing}{\spaceskip=\fontdimen2\font plus
\BIBentryALTinterwordstretchfactor\fontdimen3\font minus \fontdimen4\font\relax}
\providecommand{\BIBforeignlanguage}[2]{{%
\expandafter\ifx\csname l@#1\endcsname\relax
\typeout{** WARNING: IEEEtran.bst: No hyphenation pattern has been}%
\typeout{** loaded for the language `#1'. Using the pattern for}%
\typeout{** the default language instead.}%
\else
\language=\csname l@#1\endcsname
\fi
#2}}
\providecommand{\BIBdecl}{\relax}
\BIBdecl

\bibitem{middleton1991trade}
R.~H. Middleton, ``Trade-offs in linear control system design,'' \emph{Automatica}, vol.~27, no.~2, pp. 281--292, 1991.

\bibitem{dong2008control}
L.~Dong, Q.~Zheng, and Z.~Gao, ``On control system design for the conventional mode of operation of vibrational gyroscopes,'' \emph{IEEE Sensors Journal}, vol.~8, no.~11, pp. 1871--1878, 2008.

\bibitem{peng2005modeling}
K.~Peng, B.~M. Chen, G.~Cheng, and T.~H. Lee, ``Modeling and compensation of nonlinearities and friction in a micro hard disk drive servo system with nonlinear feedback control,'' \emph{IEEE Transactions on Control Systems Technology}, vol.~13, no.~5, pp. 708--721, 2005.

\bibitem{ouakad2015nonlinear}
H.~M. Ouakad, A.~H. Nayfeh, S.~Choura, and F.~Najar, ``Nonlinear feedback controller of a microbeam resonator,'' \emph{Journal of Vibration and Control}, vol.~21, no.~9, pp. 1680--1697, 2015.

\bibitem{luo2016chaos}
S.~Luo and Y.~Song, ``Chaos analysis-based adaptive backstepping control of the microelectromechanical resonators with constrained output and uncertain time delay,'' \emph{IEEE Transactions on Industrial Electronics}, vol.~63, no.~10, pp. 6217--6225, 2016.

\bibitem{liguo2004hybrid}
C.~Liguo, S.~Lining, R.~Weibin, and B.~Xinqian, ``Hybrid control of vision and force for {MEMS} assembly system,'' in \emph{2004 IEEE International Conference on Robotics and Biomimetics}.\hskip 1em plus 0.5em minus 0.4em\relax IEEE, 2004, pp. 136--141.

\bibitem{weller2002hysteresis}
S.~R. Weller and G.~C. Goodwin, ``Hysteresis switching adaptive control of linear multivariable systems,'' \emph{IEEE Transactions on Automatic Control}, vol.~39, no.~7, pp. 1360--1375, 2002.

\bibitem{stefanini2011miniature}
R.~Stefanini, M.~Chatras, P.~Blondy, and G.~M. Rebeiz, ``Miniature {MEMS} switches for {RF} applications,'' \emph{Journal of Microelectromechanical Systems}, vol.~20, no.~6, pp. 1324--1335, 2011.

\bibitem{zhang2016adaptive}
Y.~Zhang and Q.~Xu, ``Adaptive sliding mode control with parameter estimation and kalman filter for precision motion control of a piezo-driven microgripper,'' \emph{IEEE transactions on control systems technology}, vol.~25, no.~2, pp. 728--735, 2016.

\bibitem{clegg1958nonlinear}
J.~C. Clegg, ``A nonlinear integrator for servomechanisms,'' \emph{Transactions of the American Institute of Electrical Engineers, Part II: Applications and Industry}, vol.~77, no.~1, pp. 41--42, 1958.

\bibitem{horowitz1975non}
I.~Horowitz and P.~Rosenbaum, ``Non-linear design for cost of feedback reduction in systems with large parameter uncertainty,'' \emph{International Journal of Control}, vol.~21, no.~6, pp. 977--1001, 1975.

\bibitem{deenen2017hybrid}
D.~A. Deenen, M.~F. Heertjes, W.~Heemels, and H.~Nijmeijer, ``Hybrid integrator design for enhanced tracking in motion control,'' in \emph{2017 American Control Conference (ACC)}.\hskip 1em plus 0.5em minus 0.4em\relax IEEE, 2017, pp. 2863--2868.

\bibitem{dadkhah2025digital}
D.~Dadkhah, E.~Khodabakhshi, and S.~O.~R. Moheimani, ``Digital implementation of tracking and damping control based on hybrid integrator-gain system for a {MEMS} force sensor,'' in \emph{2025 American Control Conference (ACC)}.\hskip 1em plus 0.5em minus 0.4em\relax IEEE, 2025, pp. 4597--4602.

\bibitem{van2020experimental}
S.~Van~den Eijnden, M.~F. Heertjes, and H.~Nijmeijer, ``Experimental demonstration of a nonlinear pid-based control design using multiple hybrid integrator-gain elements,'' in \emph{2020 American Control Conference (ACC)}.\hskip 1em plus 0.5em minus 0.4em\relax IEEE, 2020, pp. 4307--4312.

\bibitem{shi2024hybrid}
K.~Shi and I.~R. Petersen, ``Hybrid integrator-gain system based integral resonant controllers for negative imaginary systems,'' in \emph{2024 IEEE 63rd Conference on Decision and Control (CDC)}.\hskip 1em plus 0.5em minus 0.4em\relax IEEE, 2024, pp. 2379--2384.

\bibitem{van2020hybrid}
S.~Van~den Eijnden, M.~F. Heertjes, W.~Heemels, and H.~Nijmeijer, ``Hybrid integrator-gain systems: A remedy for overshoot limitations in linear control?'' \emph{IEEE Control Systems Letters}, vol.~4, no.~4, pp. 1042--1047, 2020.

\bibitem{heertjes2019hybrid}
M.~Heertjes, S.~Van Den~Eijnden, B.~Sharif, M.~Heemels, and H.~Nijmeijer, ``Hybrid integrator-gain system for active vibration isolation with improved transient response,'' \emph{IFAC-PapersOnLine}, vol.~52, no.~15, pp. 454--459, 2019.

\bibitem{franklin1998digital}
G.~F. Franklin, J.~D. Powell, M.~L. Workman \emph{et~al.}, \emph{Digital control of dynamic systems}.\hskip 1em plus 0.5em minus 0.4em\relax Addison-wesley Menlo Park, CA, 1998, vol.~3.

\bibitem{sharif2022discrete}
B.~Sharif, D.~W. Alferink, M.~F. Heertjes, H.~Nijmeijer, and W.~Heemels, ``A discrete-time approach to analysis of sampled-data hybrid integrator-gain systems,'' in \emph{2022 IEEE 61st Conference on Decision and Control (CDC)}.\hskip 1em plus 0.5em minus 0.4em\relax IEEE, 2022, pp. 7612--7617.

\bibitem{shi2025discrete}
K.~Shi, E.~Khodabakhshi, P.~Biswas, I.~R. Petersen, and S.~O.~R. Moheimani, ``Discrete-time {HIGS} based digital control of negative imaginary systems,'' \emph{Control Engineering Practice}, vol. 163, p. 106386, 2025.

\bibitem{sharif2024analysis}
B.~Sharif, D.~Alferink, M.~Heertjes, H.~Nijmeijer, and M.~Heemels, ``Analysis of sampled-data hybrid integrator-gain systems: A discrete-time approach,'' \emph{Automatica}, vol. 167, p. 111765, 2024.

\bibitem{lanzon2008stability}
A.~Lanzon and I.~R. Petersen, ``Stability robustness of a feedback interconnection of systems with negative imaginary frequency response,'' \emph{IEEE Transactions on Automatic Control}, vol.~53, no.~4, pp. 1042--1046, 2008.

\bibitem{petersen2010feedback}
I.~R. Petersen and A.~Lanzon, ``Feedback control of negative-imaginary systems,'' \emph{IEEE Control Systems Magazine}, vol.~30, no.~5, pp. 54--72, 2010.

\bibitem{ghallab2025negative}
A.~G. Ghallab, M.~A. Mabrok, and I.~R. Petersen, ``Negative imaginary systems theory for nonlinear systems: A dissipativity approach,'' \emph{IEEE Transactions on Automatic Control}, vol.~70, no.~12, pp. 8120--8132, 2025.

\bibitem{shi2023output}
K.~Shi, I.~R. Petersen, and I.~G. Vladimirov, ``Output feedback consensus for networked heterogeneous nonlinear negative-imaginary systems with free-body motion,'' \emph{IEEE Transactions on Automatic Control}, vol.~68, no.~9, pp. 5536--5543, 2023.

\bibitem{shi2024discrete}
K.~Shi, I.~R.~Petersen, and I.~G.~Vladimirov, ``Discrete-time negative imaginary systems from {ZOH} sampling,'' \emph{IFAC-PapersOnLine}, vol.~58, no.~17, pp. 214--219, 2024.

\bibitem{shi2024digital}
K.~Shi and I.~R. Petersen, ``Digital control of negative imaginary systems: a discrete-time hybrid integrator-gain system approach,'' in \emph{2024 European Control Conference (ECC)}.\hskip 1em plus 0.5em minus 0.4em\relax IEEE, 2024, pp. 2611--2616.

\bibitem{shi2022negative}
K.~Shi, N.~Nikooienejad, I.~R. Petersen, and S.~O.~R. Moheimani, ``A negative imaginary approach to hybrid integrator-gain system control,'' in \emph{2022 IEEE 61st Conference on Decision and Control (CDC)}.\hskip 1em plus 0.5em minus 0.4em\relax IEEE, 2022, pp. 1968--1973.

\bibitem{shi2023negative}
K.~Shi, N.~Nikooienejad, I.~R. Petersen, and S.~O.~R.~Moheimani, ``Negative imaginary control using hybrid integrator-gain systems: Application to {MEMS} nanopositioner,'' \emph{IEEE Transactions on Control Systems Technology}, vol.~32, no.~3, pp. 718--730, 2023.

\bibitem{dadkhah2024design}
D.~Dadkhah and S.~O.~R. Moheimani, ``Design, fabrication, and control of a double-stage {MEMS} force sensor,'' \emph{IEEE/ASME Transactions on Mechatronics}, 2024.

\end{thebibliography}

\end{document}